\documentclass[12pt]{article}

\usepackage{graphicx}
\usepackage[utf8]{inputenc}

\usepackage[english]{babel}

\usepackage{a4}
\usepackage{amsmath,amssymb,amsfonts,amsthm} 
\usepackage{mathtools}
\usepackage{cite}
\usepackage{epsfig}

\newtheorem{theorem}{Theorem}[section] 

\newtheorem{lemma}[theorem]{Lemma}
\newtheorem{proposition}[theorem]{Proposition}
\newtheorem{corollary}[theorem]{Corollary}

\numberwithin{equation}{section}

\newcommand{\II}{{\mathbb I}}
\newcommand{\JJ}{{\mathbb J}}

\newcommand{\CC}{{\mathbb C}}
\newcommand{\RR}{{\mathbb R}}

\newcommand{\NN}{{\mathbb N}}

\newcommand{\cB}{{\mathcal{B}}}

\newcommand{\cF}{{\mathcal{F}}}
\newcommand{\cH}{{\mathcal{H}}}

\newcommand{\cS}{{\mathcal{S}}}

\newcommand{\oA}{{\bar{A}}}

\newcommand{\oL}{{\bar{L}}}

\newcommand{\bix}{\mbox{\boldmath $x$}}
\newcommand{\sbix}{{\mbox{\footnotesize \boldmath $x$}}}
\newcommand{\biy}{\mbox{\boldmath $y$}}

\newcommand{\biH}{\mbox{\boldmath $H$}}
\newcommand{\sbiH}{{\mbox{\footnotesize \boldmath $H$}}}

\newcommand{\biP}{\mbox{\boldmath $P$}}

\newcommand{\biQ}{\mbox{\boldmath $Q$}}

\newcommand{\biV}{\mbox{\boldmath $V$}}

\newcommand{\bfA}{\mbox{\boldmath $\mathfrak A$}}
\newcommand{\obfA}{\overline{\mbox{\boldmath $\mathfrak A$}}}

\newcommand{\bfAal}{\mbox{\boldmath $\mathfrak A$}_{\, \balpha_L}}

\newcommand{\bfC}{\mbox{\boldmath $\mathfrak C$}}

\newcommand{\bfCal}{\mbox{\boldmath $\mathfrak C$}_{\, \balpha_L}}

\newcommand{\bfI}{\mbox{\boldmath $\mathfrak I$}}

\newcommand{\bfR}{\mbox{\boldmath $\mathfrak R$}}

\newcommand{\balpha}{\mbox{\boldmath $\alpha$}}
\newcommand{\sbalpha}{\mbox{\footnotesize \boldmath $\alpha$}}

\newcommand{\bSigma}{\mbox{\boldmath $\Sigma$}}

\newcommand{\bnabla}{\mbox{\boldmath $\nabla$}}
\newcommand{\blk}{{\boldsymbol [ }}
\newcommand{\brk}{{\boldsymbol ] }}
  
\newcommand{\Adl}{\text{Ad} \, e^{itH_L}}
\newcommand{\Ad}{\text{Ad} \, e^{itH}}
\newcommand{\ad}[1]{\text{Ad} \, e^{i #1 \sbiH_{0,n}} }

\def\ie{{\it i.e.\ }}

\newcommand{\be}{\begin{equation}}
\newcommand{\ee}{\end{equation}} 

\begin{document}

\title{Resolvent algebras and limit states of \\
  interacting canonical ensembles}
\author{Detlev Buchholz \\[2mm]
\small Mathematisches Institut, Universit\"at G\"ottingen, \\
\small Bunsenstr.\ 3-5, 37073 G\"ottingen, Germany \\ 
\small detlev.buchholz@mathematik.uni-goettingen.de \\[5pt]
}
\date{}

\maketitle

\noindent \textbf{Abstract.} The limit states of canonical ensembles of a 
large number of interacting bosons 
at a given temperature, which are confined by harmonic forces, are studied in
the framework of the resolvent algebra. It is shown that the limits satisfy
the KMS condition or are ground states,
regardless of the type of interaction. In case of attractive
forces, where the ensembles collapse, observables that become meaningless
in the limit disappear from the limit representations.
For repulsive forces, this
can also happen if condensates with an infinite number of particles
in the same state (proper condensates)
appear in the limit. The resulting structures and their interpretation
are illustrated by a simple model.
The study of vanishing harmonic forces
(thermodynamic limit) involves changes of the dynamics. It is
conveniently based on derivations acting on the algebra. They are
given by the commutator of the Hamiltonians with the elements of
the algebra. To ensure that the images remain in the
algebra, the interaction must be regularized. 
This is accomplished in a manner that has only a minor impact on
the dynamics and may be of broader 
interest. With this input a relation between the strength
of the confining harmonic forces and the number of particles in the ensembles
is derived from the condition that the limit states are to be stationary
(invariant) 
under the adjoint action of the unconfined, spatially homogeneous
limit dynamics. This relation encompasses 
the conditions that are frequently used in studies of Bose-Einstein
condensates. 

\medskip  \noindent
\textbf{Keywords} \  resolvent algebra $\cdot$ interacting bosons $\cdot$
limits of equilibrium states 

\medskip  \noindent
\textbf{Mathematics Subject Classification} \
 46L60 $\cdot$ 81V73 $\cdot$  82B10 

\section{Introduction}
\label{sec1}
\setcounter{equation}{0} 

Resolvent algebras were introduced in \cite{BuGr} as an 
algebraic framework for the study of many body quantum systems. The
framework was extended in \cite{Bu1,Bu2} to non-relativistic
Bose fields in order to cover systems with an arbitrary number
of bosons, including the case of an infinite number. In
contrast to other algebraic approaches, the resolvent algebras are
stable under the action of a large family of dynamics and include
many observables of physical interest. It is the aim of the present article
to show that they shed new light on the properties
of infinite systems. 

\medskip
We begin by recalling some basic notions and results from \cite{Bu1}.
The resolvent algebra of Bose fields is generated by the resolvents
of the field operators, \mbox{$(i \lambda + a^*(f) + a(f))^{-1}$},
$\lambda \in \RR \backslash \{0\}$, 
where $a^*(f), a(f)$ are creation and annihilation operators 
satisfying canonical commutation relations
\be
[a(f), a^*(f)] = \langle f | f \rangle \, 1 \, . 
\ee
They are integrated with test functions 
$f \in \cS(\RR^s)$, where $\RR^s$ denotes space. This choice of
test functions is convenient in the present investigation, 
but any other dense subspace of~$L^2(\RR^s)$ would work.  
The resolvents act on Fock space $\cF$, which is the direct sum of
$n$-particle spaces $\cF_n$, $n \in \NN_0$, the one-particle space $\cF_1$
being identified with~$L^2(\RR^s)$. The completion of this algebra
with regard to
the operator norm on $\cF$ is a unital
C*-algebra, the resolvent algebra $\bfR$. 

\medskip
On the algebra $\bfR$ act the automorphisms generated by the
unitary exponentials 
of the particle number operator $N$, called gauge transformations. 
The subalgebra of all gauge invariant elements of $\bfR$, the algebra
of observables, is denoted by $\bfA$. It is generated by the
resolvents, cf.\ \cite{BuYn}, 
\be
R_\mu(f) \coloneqq (\mu + a^*(f)a(f))^{-1} \, , \quad \mu > 0, \,
f \in \cS(\RR^s) \, .  
\ee
The algebra $\bfA$ is faithfully represented on $\cF$ and leaves all
$n$-particle spaces $\cF_n$ invariant. Its representations
$\rho_n$ on $\cF_n$, however, have non-trivial kernels, $n \in \NN_0$.
It implies that the algebra $\bfA$ contains non-trivial two-sided 
ideals. This feature is a necessary prerequisite for describing dynamics
involving interactions, cf. \cite{BuGr}, and it also matters in the present
discussion. 

\medskip 
In order to incorporate the action of dynamics,
the algebra $\bfA$ has to be extended.
This extension is based on the observation that there
exist surjective *-homomorphisms mapping the algebras
$\rho_{n+1}(\bfA)$ onto $\rho_n(\bfA)$, $n \in \NN_0$, cf.\ \cite{Bu1}.
The corresponding
inverse limit of this family of algebras is
again a C*- algebra~$\obfA$,
the (uniformly bounded) projective limit of
the coherent sequence of algebras $\rho_n(\bfA)$, $n \in \NN_0$.
In the present investigation we make use of the fact that the elements  
of $\obfA$ are all bounded operators
on $\cF$ whose restrictions to $\sum_{m = 0}^n \cF_m$, $n \in \NN_0$,
coincide with $\oplus_{m=0}^n \rho_m(A)$ for some $A \in \bfA$;  
the operator $A$ depends in general on~$n$, however.

\medskip 
Special
elements of $\obfA$ are the resolvents $(\mu + a^*(f)a(f))^{-1}$ for
arbitrary $f \in L^2(\RR^s)$, $\mu > 0$. Even if they are not
contained in $\bfA$, their representations on $\cF_n$ are contained
in $\rho_n(\bfA)$. In fact, there exist projections (surjective
\mbox{*-homomorphisms})
mapping $\obfA$ onto $\rho_n(\bfA)$ \cite[Sect.\ 1]{Ph}; thus the elements
$\oA \in \obfA$ correspond to coherent sequences $\oA_n \in \rho_n(\bfA)$,
$n \in \NN_0$. By some slight abuse of notation, we denote these projections
also by $\rho_n$ and in this way extend the representations of $\bfA$
to representations of $\obfA$ on $\cF_n$. Similarly,
we retain the notation $\omega_n$  for $n$-particle
states extended from
$\bfA$ to $\obfA$, $n \in \NN_0$. The operator norm on $\cF_n$ is denoted
by $\| \, \cdot \, \|_n$. Since the representations $\rho_n(\oA)$ of 
$\oA \in \obfA$   coincide with the restriction of $\oA$ to $\cF_n$, one has
$\| \rho_n(\oA) \|_n = \| \oA \|_n$, $n \in \NN_0$. 

\medskip
Within this setting we consider Hamiltonians on $\cF$ of the form
\begin{align}  \label{e.1.2}
  \biH_L \coloneqq \int \! d\bix \, (\nabla a^*(\bix) \nabla a(\bix) +
  L^{-4} \, \bix^2 a^*(\bix) a(\bix) ) \nonumber \\
  + \int \! d\bix \, d\biy \, a^*(\bix) a^*(\biy) V(\bix - \biy) a(\bix)
  a(\biy) \, , 
\end{align}  
where $L > 0$ and $V \in C_0(\RR^s)$ is a continuous two-body potential
vanishing at infinity. We restrict our attention here to this regular
family. But unbounded and singular potentials could be 
admitted; examples are given in \cite[Sec.\ 6]{Bu3}.
The Hamiltonians $\biH_L$ are gauge invariant and their restrictions
to $\cF_n$ coincide with the $n$-particle Hamiltonians, $n \in \NN$,
\be \label{e.1.3}
\biH_{L,n} \coloneqq \sum_{k=1}^n (\biP_k^2 + L^{-4} \, \biQ_k^2)
+ \sum_{l \neq m,1}^n V(\biQ_l - \biQ_m) \, ,
\ee
where $\biQ_k, \biP_k$ are pairs of canonically conjugate position and
momentum operators, \mbox{$k = 1, \dots , n$}. In the
thermodynamic limit, where~$L$ tends to infinity, the Hamiltonians $\biH_L$ 
approach in the strong resolvent sense
the spatially homogeneous
Hamiltonian without harmonic potential, which is  denoted by $\biH$.

\medskip
Since the Hamiltonians are selfadjoint, one can proceed to the
corresponding unitary
time translations that define automorphisms
$\balpha_L(t) \coloneqq \Adl$, respectively
$\balpha(t) \coloneqq \Ad$, $t \in \RR$, of the algebra of
bounded operators $\cB(\cF)$ on Fock space $\cF$. In
  what follows, the homogeneous dynamics $\balpha$ is 
  identified with $\balpha_L$ for $L = \infty$.
In general, these automorphisms do not leave the subalgebra
$\bfA \subset \cB(\cF)$ invariant.
But, as has been shown in \cite{Bu1}, the
represented algebras $\rho_n(\bfA)$ on $\cF_n$ are 
invariant under their action, $n \in \NN_0$. Moreover, the
*-homomorphisms relating these algebras intertwine this action. It
follows that $\balpha(t)(\bfA) \subset \obfA$, $t \in \RR$,  
and one has 
\be
\rho_n(\alpha_L(t)(A)) = \mbox{Ad} \, e^{it\sbiH_{L,n}} \rho_n(A) \, , \quad
t \in \RR \, , \ A \in \bfA \, ,
\ee
\ie the representations are covariant. 
The C*-algebra generated by the operators $\balpha_L(t)(A)$ with
$t \in \RR$, $A \in \bfA$, is denoted by $\bfAal$. It is the
natural extension of the kinematic algebra $\bfA$ to a dynamical
algebra on which the dynamics $\balpha_L$ acts. 

\medskip
The time translations act pointwise norm
continuously on the represented algebras
$\rho_n(\bfAal)$, $n \in \NN_0$, but this is
not so on $\bfAal$. For given dynamics, this
property can be achieved by regularizing the action of the
automorphisms~$\balpha_L(\RR)$ on
$\bfAal$. It is accomplished by integration of the time
translated operators with test functions. As will be shown, the
result is a subalgebra $\bfCal \subset \obfA$  on
which the respective automorphisms $\balpha_L(\RR)$ act
in the desired norm continuous manner. It yields for any given dynamics a 
C*-dynamical system which describes finite and infinite numbers
of Bosons. 

\medskip
In the subsequent section, we focus on stationary (time
translation invariant) states
of the C*-dynamical systems $(\bfCal, \balpha_L)$.
We will show that all states in the
corresponding GNS representations can continuously be extended to the
algebra~$\bfAal$. Hence, no information
about the underlying observables is lost 
by the regularization. Basic examples of stationary states are
the canonical Gibbs ensembles of $n$ particles at inverse temperature
$\beta > 0$. As we shall see, these ensembles have,
for large $n$, limits which satisfy the Kubo-Martin-Schwinger (KMS) condition.
It is the defining property of equilibrium for infinite systems 
\mbox{\cite[Ch.\ 5.1]{Ha}}. Similarly, the limits of ground states are 
ground states. Analogous results hold for grand canonical ensembles, 
but they are not discussed here.

\medskip 
These results hold
independently of the type of interaction. This may be surprising 
since in case of attractive forces the ensembles are expected to
collapse, which seems to be in conflict with equilibrium. The explanation 
is based on  the fact that observables disappear, which become non-nonsensical
in the limit, \ie they move into the kernel of the
resulting representation. Since the kernel of a representation is an
ideal of the underlying algebra, it shows that the non-trivial
ideal structure of $\bfA$ matters also here. 

\medskip
In case of repulsive forces or two-body potentials of positive type, there
appear in general Bose-Einstein condensates, \ie clusters of bosons
being in the same single particle state. That state
may be infinitely
occupied in the limit, \ie form a proper condensate \cite{Bu4}. The
particle density is then not finite, so the limit ensemble is
not locally normal with regard to the Fock representation.
This local normality is often taken as a selection criterion for
states of physical interest and is enforced by rescaling the
external potential, cf.\ \cite[Ch.\ 6]{LiSeSoYn}.
Here, the present algebraic point of view 
provides a more refined picture.

\medskip
The appearance of proper condensates is virtually inevitable if the
number of particles is increased arbitrarily while keeping the
external potential fixed. 
As we shall see, this fact implies
that there exists a two-sided ideal $\bfI \subset \bfAal$
which lies in the kernel of the limit representation, so the
corresponding observables disappear. Only the   
observables in the quotient $\bfAal / \bfI$ can be
tested in the limit state. Hence excitations of the
condensate are accessible to observations, but not its individual 
members. This is 
reminiscent of the Dirac sea, which is filled with an infinite
number of fermions, the existence of which can only be detected
by their excitations. It suggests that ground states,  
containing a proper condensate of bosons, may be interpreted   
as a confined vacuum. We will illustrate this idea by a
simple non-interacting example. 

\medskip
Another important idealization is the thermodynamic limit, 
which, in the present setting, amounts to sending the
confinement length $L$ of the
harmonic potential to infinity. So one changes the dynamics.
It has been shown in 
\cite[App.]{Bu1} that the automorphisms $\balpha_L(t)$
of $\bfA$ converge in the limit of large $L$ pointwise in norm
to $\balpha(t)$, $t \in \RR$, in all representations $\rho_n$, $n \in \NN_0$.
We present here an improved version of this result. It 
implies that for any 
particle number $n$ one can choose a length $L_n$
such that the canonical ensembles 
of the dynamical systems $(\bfC_{\, \balpha_{L_n}}, \balpha_{L_n})$ have 
limit points for large $n$ that are invariant under the homogeneous
dynamics $\balpha$. This holds provided the
confinement lengths approach infinity sufficiently rapidly,
\ie the harmonic potential becomes flat fast enough.

\medskip
In order to explore the properties of the limit states, such as 
the particle density, one needs some more detailed information about
the dependence of $L_n$ on $n$. Instead of trying to improve the estimates
in \cite{Bu1}, we consider here the generators of the dynamics,
\ie the corresponding derivations that are  
given by the action of the commutators of the Hamiltonians 
with the elements of their domains in the algebras
$\rho_n(\bfA)$, $n \in \NN_0$. 
Within the present setting, the
basic resolvents form a natural domain for the derivations.
Yet in order to ensure that their images are contained in the 
algebras, one must regularize the interaction. This is accomplished in a manner
which has only a minor impact on the dynamics and may be of more
general interest. Namely,
the regularized and non-regularized unitary time translations can be
chosen to be arbitrarily close in norm for any given compact time
interval. It justifies to consider the \mbox{regularized}
derivations. With this input, we determine the dependence $n \mapsto L_n$
which implies that the limit states lie in the kernel of the limit
derivation, generating the homogeneous dynamics.
This is a meaningful test for the stationarity of the limit states.
We find that this condition is satisfied whenever $n / L_n^8 $
approaches $0$ in the limit of large~$n$. It applies to 
all types of interactions and encompasses the condition
$n / L_n^6 = \text{const.}$ that is frequently used in case of 
harmonically trapped Bosons~\cite{LiSeSoYn}. 

\medskip
Our article is organized as follows. In Section 2, we outline
the construction of C*-dynamical systems and
discuss the properties of confined canonical ensembles that 
arise in the limiting case of infinite particle numbers.
Section 3 contains a discussion of the thermodynamic limit of
the ensembles, in particular  the relationship between the
number of bosons and the confinement length 
that leads to stationary limit states. The article
closes with brief conclusions. In the appendix, the proofs
of two lemmas are given. 

\section{Limit states of confined ensembles}
\label{sec2}
\setcounter{equation}{0} 

In this section we study the properties of confined canonical ensembles,
including ground states, in the limit of large particle numbers. Thus we
keep the dynamics~$\balpha_L$ fixed and consider states with an
increasing number of particles on the algebra~$\bfAal$, on which the dynamics
acts. This action is not pointwise continuous in the norm
topology, which is a
desirable feature in the analysis of the limits. We therefore
regularize the
operators: Given any operator $A \in \bfAal$ and any test function
$\varphi \in \cS(\RR)$, we define
\be \label{e.2.1} 
\balpha_L(\varphi)(A) \coloneqq \int \! ds \, \varphi(s)
\balpha_L(s)(A) \, .
\ee
Since the automorphisms are defined by the adjoint action of a continuous
unitary representation of $\RR$, the integral exists in the strong operator
topology on Fock space $\cF$. Transferring the action of the 
dynamics to $\varphi$ by a change of variables, one finds that the function 
$t \mapsto \balpha_L(t)\big(\balpha_L(\varphi)(A)\big)$ is smooth 
in the norm topology. These smooth elements generate a C*-algebra
$\bfCal$ on which the dynamics acts pointwise norm
continuously. It follows from the subsequent lemma that $\bfCal$ 
is a subalgebra of $\obfA$.
\begin{lemma} Let $A \in \bfAal$ and $\varphi \in \cS(\RR)$. 
  The regularized operator $\balpha_L(\varphi)(A)$ is an element of $\obfA$.
\end{lemma}  
\begin{proof}
  The bounded operator $\balpha_L(\varphi)(A)$ is defined on $\cF$
  as a Riemann integral
  in the strong operator topology. It can be approximated in this
  topology by
  a sequence of Riemann sums $S_k \in \bfAal$ involving the
  time translated operator \mbox{$A \in \bfA$}, $k \in \NN$. These sums
  are elements of $\obfA$. 
  The representatives $\rho_n(S_k)$ of these sums converge on $\cF_n$ 
  in norm to $\rho_n(\balpha_L(\varphi)(A)) \in \rho_n(\bfA)$ for large
  $k$,  
  because the dynamics $\balpha_L(\RR)$ acts norm continuously 
  on $A$ within this representation. This applies to all
  representations $\rho_n$, $n \in \NN_0$, proving that 
  \mbox{$\balpha_L(\varphi)(A) \in \obfA$}. 
 \end{proof}  

The regularization simplifies the analysis of states on $\bfCal$
that are invariant under the adjoint action of the dynamics $\balpha_L$.
To study their properties, 
we need to consider $k$-fold products of time translated operators
in $\bfAal$ which are jointly integrated with test functions
in $\cS(\RR^k)$, $k \in \NN$. As a consequence of the kernel theorem for
Schwartz test functions, the resulting operators are contained
in the algebra $\bfCal$, which is generated by sums of products of individually
regularized operators. However, by the regularization some relevant observables
get lost. Examples are the basic resolvents, which 
determine the number of particles in a particular one-particle state.
On them the time translations do not act norm
continuously. Yet these observables can be recovered in the
GNS representations of the regularized algebra, as is shown next.
\begin{lemma}  \label{l.2.2}
  Let $\omega$ be a state on $\bfCal$
  which is invariant under the adjoint action
  of $\balpha_L(\RR)$, and let $(\pi,\cH,\Omega)$ be its GNS representation.
  This representation is covariant, \ie there
  exists a continuous unitary representation $U_\pi$ of $\RR$ such that
  \be
  \text{Ad} \, U_\pi(t)(\pi(C))  = \pi(\balpha_L(t)(C)) \, , \quad
  t \in \RR \, , \, C \in \bfCal \, . 
  \ee
  Moreover, $\pi$ extends weakly continuously to
  $\bfAal$ as a
  positive, unit-preserving, and covariant bi-module map $\overline{\pi}$.
  That is,  given  $C',C''  \in \bfCal$ and $A \in \bfAal$, one has 
  $\overline{\pi}(C' \, A C'') = \pi(C') \overline{\pi}(A)
  \pi(C'')$.
\end{lemma}
\noindent {\bf Remark:} It follows from this result that all 
states in the GNS-representation of $\bfCal$ can be extended weakly
continuously to states on $\bfAal$.
\begin{proof}
  The existence of $U_\pi$ follows from standard arguments,
  cf.\ \cite[Ch.\ II.3.2]{Ha}: one puts
  \be
  U_\pi(t)\, \pi(C) \Omega \coloneqq \pi(\balpha_L(t)(C)) \Omega \, , \quad
  t \in \RR \, , \ C \in \  \bfCal \, .
  \ee
  The unitarity of $U_\pi$ 
  results from the invariance of $\omega$ under the adjoint 
  action of $\balpha_L(\RR)$ and its continuity from
  the pointwise norm continuous
  action of $\balpha_L(\RR)$ on $\bfCal$.

  \medskip
  For the proof that $\pi$ extends
  to $\bfAal$, one makes use of the fact that 
  $\bfCal$ contains the operators
  $\balpha_L(\varphi)(C' A C'')$, where 
  $\varphi \in \cS(\RR)$, $C', C'' \in \bfCal$, and $A \in \bfAal$. 
  Since $\bfCal$ is generated by regularized elements of $\bfAal$,
  this statement follows from the kernel theorem for Schwartz
  test functions. We use the
  abbreviated notation $\oA \coloneqq C' A C''$. 
  Let $\delta_k$, $k \in \NN$,
  be a sequence of positive test functions with integral $1$
  that approaches the
  Dirac delta-function. For given $\oA$, the sequence
  $\balpha_L(\delta_k)(\oA) \in \bfCal$, $k \in \NN$,
  is uniformly bounded. Hence there
  are sub-sequences $\balpha_L(\delta_i)(\oA)$, $i \in \II \subset \NN$, 
  such that $i \mapsto \pi(\balpha_L(\delta_i)(\oA))$ converges in the
  weak operator topology to some operator $A_\II \in \cB(\cH)$.
  Let $\varphi \in \cS(\RR)$ be a test function. Using 
  the fact that $U_\pi$ is a unitary representation of
  $\RR$, one obtains for $i \in \II$ 
  \be
  \int \! ds \, \varphi(s) \text{Ad} \, U_\pi(s) (\pi(\balpha_L(\delta_i)(\oA)))
  = \int \! ds \, \delta_i(s) \text{Ad} \, U_\pi(s)
  (\pi(\balpha_L(\varphi)(\oA))) \, .
  \ee
  Because of the norm continuity of the action of
  $s \mapsto \text{Ad} \, U_\pi(s)$
  on $\pi(\bfCal)$, the right hand side of this equality converges in
  norm to $\pi(\balpha_L(\varphi)(\oA))$ for any choice of
  the sub-sequence $\II$. Since the map
  $\int \! ds \, \varphi(s) \text{Ad} \,
  U_\pi(s) \, : \cB(\cH) \rightarrow  \cB(\cH)$
  is \mbox{$\pi$-normal}, the left hand side converges
  for  $i \in \II$ weakly to
  $\int \! ds \, \varphi(s) \text{Ad} \, U_\pi(s) (A_{\II})$. So one has
  \be
  \int \! ds \, \varphi(s) \text{Ad} \, U_\pi(s)(A_{\II}) =
  \pi(\balpha_L(\varphi)(\oA)) \, .
  \ee
  Let $A_\JJ$ be the weak limit of another sequence
  $\pi(\balpha_L(\delta_j)(\oA))$, $j \in \JJ \subset \NN$. It follows
  that for all $\varphi \in \cS(\RR)$
  \be
  \int \! ds \, \varphi(s) \text{Ad} \,  U_\pi(s)(A_{\II})
  =  \pi(\balpha_L(\varphi)(\oA))
  = \int \! ds \, \varphi(s) \text{Ad} \,  U_\pi(s)(A_{\JJ}) \, .
  \ee
  Because of the weakly continuous action of
  $s \mapsto \text{Ad} \, U_\pi(s)$ on $\cB(\cH)$,
  it implies $A_\II = A_\JJ$. Thus, all weak limit
  points of the sequence $k \mapsto \pi(\balpha_0(\delta_k)(\oA))$
  coincide, \ie the sequence converges weakly. Recalling that
  $\oA =  C'AC''$, we put
  \be \label{2.7}
  \overline{\pi}(C' A C'') \coloneqq
  \lim_k \pi(\balpha_L(\delta_k)(C' A C'')) \, ,
  \quad C', C'' \in \bfCal \, , \, A \in \bfAal \, .
  \ee
  It is obvious that $\overline{\pi}$ is linear, unit preserving, and it 
  maps positive operators to positive ones. It is also clear from its
  definition that it is covariant.

  \medskip
  For the proof that $\overline{\pi}$ is a bi-module map, we use of
  the fact that $s \mapsto \balpha_L(s)(C')$ is norm continuous (similarly for
  $C''$). Denoting the operator norm on $\obfA$
  by $\| \, \cdot \, \|$, one has
   \begin{align}
     &     \| \balpha_L(\delta_k)(C' A C'') - C \, \balpha_L(\delta_k)(A) C'' \|
     \nonumber \\
&     = \| \int \! ds \, \delta_k(s) \big(
     \balpha_L(s)(C' A C'') - C' \, \balpha_L(s)(A) C'' \big) \|
     \nonumber \\
& \leq \int \! ds \, \delta_k(s) \| (\balpha_L(s)(C') - C') \| \,
     \| A \| \, \| C'' \|
     \nonumber \\
&     + \int \! ds \, \delta_k(s) \| (\balpha_L(s)(C'') - C'') \| \,
   \| C' \| \, \| A \| \, .
   \end{align}  
   This upper bound tends to $0$ in the limit of large $k$. Hence, in the
   sense of weak convergence,
   \begin{align}
  \overline{\pi}(C' A C'') & =  \lim_k \pi(\balpha_L(\delta_k)(C' A C'')) 
  =
  \pi(C')  \lim_k \pi(\balpha_L(\delta_k)(A)) \, \pi(C'') \nonumber \\
  & =
      \pi(C') \, \overline{\pi}(A) \, \pi(C'') \, ,
   \end{align}  
completing the proof. \end{proof}

We turn now to the discussion of the limit
states of canonical ensembles for the dynamics $\balpha_L$. Given the
inverse temperature $\beta \in \RR_+ \cup \{ \infty \}$ and
the number of particles $n \in \NN_0$, 
we consider the states formed by Gibbs-von Neumann ensembles,  
\be
\omega_{\beta,n}(A) \, \coloneqq \, \text{Tr}_{\cF_n\,}
e^{- \beta \sbiH_{L,n}} A \, / \,
\text{Tr}_{\cF_n \,}  e^{- \beta \sbiH_{L,n}} \, ,
\quad A \in \bfCal \, ,
\ee
where for $\beta = \infty$ these are the ground states
of the Hamiltonians. 
The gauge invariant operators $e^{- \beta \sbiH_L}$ and $A$ 
are restricted here to the $n$-particle
spaces $\cF_n$, where the corresponding traces are taken. 
Because of the confining potential and the boundedness of the interaction
potentials on $\cF_n$, $n \in \NN_0$, the ground states are unique
and the traces exist. However, the
right hand side of this equality is ill defined if $n$
tends to infinity. One then uses the fact that equilibrium
states satisfy the KMS condition, which is the characteristic 
property of equilibrium for finite and infinite systems
\cite[Sec.\ V.1]{Ha}.

\medskip \noindent
\textbf{Definition:} Let $\omega_\beta$ be a state on $\bfCal$.
It satisfies
the KMS condition at inverse temperature $\beta > 0$ for the dynamics
$\balpha_L$ if for all $ A,B \in \bfCal$ the
correlation functions 
\be
t \mapsto \omega_\beta(A \balpha_L(t)(B)) \, , 
\ee
can continuously be extended to the strip
$S_\beta \coloneqq \{ z \in \CC : 0 \leq \text{Im} \, z \leq \beta \}$,
are analytic in its interior, and have at its upper rim
$\text{Im} \, z = \beta$ the boundary value 
\be
t \mapsto \omega_\beta(\balpha_L(t)(B) A) \, .
\ee
If $\beta = \infty$, this boundary condition is replaced
by the condition that the analytic continuations are uniformly 
bounded on the whole upper half plane. 

\bigskip
After these preparations we can draw on general results for
sequences of KMS states
on C*-dynamic systems. We make use of the fact that for any
sequence of states on a
C*-algebra there exist states which are limits
of the sequence in the weak-*-topology. That is, for any
given finite set of operators in the algebra there exists a
sub-sequence of states such that the corresponding expectation values 
converge to the expectation values in the limit state.
Here, we deal with sequences of normal states $\omega_n$ on the algebra
$\bfA$ with respect to the representations $\rho_n$, $n \in \NN_0$.
As explained, they can be extended to states on the C*-algebra $\obfA$.
This is why their w-*-limit points also define states on that algebra. 
We can apply now \cite[Prop.\ 5.3.23]{BrRoII} to the systems at hand. 
\begin{theorem} \label{t.2.3} 
  Let the dynamics $\balpha_L$ and the
  inverse temperature \mbox{$\beta \in \RR_+ \cup \infty$} be given and 
  let $\omega_{\beta,n}$ be the sequence of states on $\obfA$
  describing the canonical Gibbs-von Neumann ensembles for the
  particle numbers $n \in \NN_0$. 
  Every weak-*-limit point of this
  sequence is a KMS state, respectively ground state,
  on $\bfCal \subset \obfA$ for the 
  inverse temperature~$\beta$. 
\end{theorem}
\noindent \textbf{Remarks:} (1) The limit states may not be unique, \ie the
sequence need not converge in general. There can be several limit
points, for example at phase transition points. 

\medskip 
\noindent (2) All normal states in the GNS representation of a limit
state can be extended continuously from $\bfCal$
to the dynamical algebra $\bfAal$, cf.\
Lemma \ref{l.2.2}. The resulting correlation 
functions for different times may not be continuous,
however. 

\medskip
As already mentioned, this theorem applies to all two-body interaction
potentials in $C_0(\RR^s)$, irrespective of the attractive or repulsive
nature of the resulting forces or their range. For the analysis
of the properties of the limit states, such as the
occurrence of Bose–Einstein condensates,
the basic resolvents are a useful tool. The following lemma
establishes one of their key properties.
\begin{lemma} \label{l.2.4}  
  Given $Z \coloneqq (\mu + a^*(f)a(f))^{-1} \in \bfA$ 
  with $\mu > 0$, $f \in \cS(\RR^s)$, let $\blk \bfA \, Z \brk$ 
  and $\blk Z \, \bfA \brk$ be the
  closed left and right ideals in $\bfA$, generated by $Z$. Then 
  $ \blk \bfA \, Z \brk = \blk Z \, \bfA \brk$. 
  Likewise, for the closed left and right ideals
  $\blk Z \, \obfA \brk$, ~$\blk Z \, \obfA \brk$ in $\obfA$,
  generated  by $Z$, it  holds
  that $ \blk \obfA \, Z \brk = \blk Z \, \obfA \brk$ .
\end{lemma}
\begin{proof}
  For each resolvent
  $R_\nu(g) \coloneqq (\nu + a^*(g)a(g))^{-1}$ with $\nu > 0$, $g \in \cS(\RR^s)$, 
  there exists an operator $A_g \in \bfA$ such that
  $Z R_\nu(g) = A_g Z$. (Since $Z$ and $R_\nu(g)$ are selfadjoint, the 
  relation $R_\nu(g) Z = Z A_g^*$ then obtains as well.)
  The operator $A_g$ is formally given by
  \be
  A_g = R_\nu(g) + Z \, [R_\nu(g) , Z^{-1} ] = R_\nu(g) + 
  Z \, [R_\nu(g) , a^*(f)a(f) ] \, .
  \ee
  Making use of the canonical commutation relations, one obtains 
  \be \label{e.2.15}
  [ R_\nu(g) , a^*(f) ] 
    =   - \langle g | f \rangle \, R_\nu(g) a^*(g) R_\nu(g)  \, , 
  \ee 
  where $\langle \cdot \, | \, \cdot \rangle$ denotes the scalar product on
  $L^2(\RR^s)$, and taking the adjoint a similar relation obtains for the
  commutator with $a(f)$.
  The operators $R_\nu(g) a^*(g)$ and $ R_\nu(g) a(g)$ are elements of
  the resolvent algebra $\bfR$. This is a consequence of the fact that
  the operators 
  \begin{align}
    &  R_\nu(g) \, (a^*(cg) + a(cg))\big(i + \varepsilon (a^*(cg) + a(cg))\big)^{-1}
    \nonumber \\
    &  = R_\nu(g)
    (1/\varepsilon) \big(1 - i \big(i + \varepsilon
    (a^*(cg) + a(cg))\big)^{-1} \big)
  \end{align}
  are, for $\varepsilon > 0$ and $c \in \CC$, elements of $\bfR$. Due to the
  damping effect of the factor $R_\nu(g)$, these operators   
  converge for small $\varepsilon$ in norm to
  $R_\nu(g) \, (a^*(cg) + a(cg))$. It implies that
  the creation and annihilation operators
  $a^*(g), a(g)$, act as derivations on the polynomial algebra,
  generated by the basic resolvents in $\bfA$; 
  their image sets lie in $\bfR$. Applying the relations involving
  these derivations a few times, one arrives~at
  \begin{align}
    A_g & = \ R_\nu(g) + \langle f | g \rangle \,  Z a^*(f) \,
    R_\nu(g) \, a(g) \, R_\nu(g) -
    \langle g | f \rangle \, Z a(f) \, R_\nu(g) \, a^*(g) \, R_\nu(g)
    \nonumber \\
    & + | \langle f | g \rangle |^2  Z  R_\nu(g)\big( 1 
    - a(g) \, R_\nu(g) \, a^*(g) 
    - a^*(g) \, R_\nu(g) \, a(g) \big) R_\nu(g) \, . 
  \end{align}
  By the preceding remarks, this operator is an element of
  $\bfR$. Since it is gauge invariant, it is contained in
  the subalgebra $\bfA$. 
  Now the algebra $\bfA_0 \subset \bfA$ that is generated by all finite sums
  and products of the basic resolvents is norm dense in $\bfA$. By the
  above computations, one has $Z \bfA_0 \subset \bfA Z$, whence, for
  their norm closures, $\blk Z \bfA \brk \subset \blk \bfA Z \brk $. 
  By taking adjoints, one also has $ \blk \bfA Z \brk \subset
  \blk Z \bfA \brk $, proving the first part of the statement.

  \medskip
  For the proof of the second part we proceed to the projective limits
  $\overline{\blk \bfA Z \brk}$ and $\overline{\blk Z \bfA \brk}$
  of these ideals. The elements $\oL \in \overline{\blk \bfA Z \brk}$
  are determined by coherent, bounded sequences
  $\rho_n(\oL) \in \rho_n(\overline{\blk \bfA Z \brk})$, $n \in \NN_0$. 
  Given any $\oA \in \obfA$ with coherent, bounded sequence
  $\rho_n(\oA) \in \rho_n(\bfA)$, $n \in \NN_0$, the  product $\oA \oL$
  is given by the sequence $\rho_n(\oA \oL) = \rho_n(\oA) \rho_n(\oL)
  \in \rho_n(\bfA) \rho_n(\blk \bfA Z \brk) =
  \rho_n(\blk \bfA Z \brk)$, $n \in \NN_0$. Thus 
  $\overline{\blk \bfA Z \brk}$  is stable under left multiplication
  by $\obfA$ and the analogous statement obtains for the right ideal
  $\overline{\blk Z \bfA \brk}$. Since $Z \in \obfA$, it shows that
  $\overline{\blk \bfA Z \brk} = \blk \obfA Z \brk$, the left ideal
  in $\obfA$ determined by $Z$; being the projective limit
  of a closed subspace
  of $\bfA$, it is closed. Similarly, one has
  $\overline{\blk Z \bfA \brk} = \blk Z \obfA \blk$. From 
  these relations and the preceding results, we obtain the equalities
  \be
  \blk \obfA Z \brk = \overline{\blk \bfA Z \brk}
  = \overline{\blk Z \bfA \brk} = \blk Z \obfA \brk \, ,
  \ee 
  which completes the proof. 
 \end{proof}
This lemma provides the basis for the following proposition, which shows that
the basic resolvents play a prominent role in the central decomposition of
states on $\obfA$ into primary (factor) states, describing 
pure thermodynamic phases. 
\begin{proposition} \label{p.2.5} 
  Let $\omega$ be a state on $\obfA$ with corresponding GNS representation
  $(\pi,\cH,\Omega)$ and let $Z_\mu(f) \coloneqq (\mu + a^*(f)a(f))^{-1}$, 
  where $\mu>0$, $f \in \cS(\RR^s)$. The limit 
  \be
  P(f) \coloneqq
  \lim_{\mu \rightarrow \infty} \pi(\mu \, Z_\mu(f)) 
  \ee
  exists in the strong operator topology. The operator
  $P(f)$ is a selfadjoint projection. It  satisfies 
  $(1 - P(f))  \pi(Z_\nu(f)) = 0$, $\nu > 0$,
  and is contained in the center of the weak closure
  $\pi(\obfA)^{''}$.
 \end{proposition}   
  \begin{proof}
    The sequence of positive
    operators $\mu \mapsto \mu \, Z_\mu(f)$ is
    increasing with $\mu$ and bounded by $1$. So it converges in the
    strong operator topology to~$P(f)$ in the representation $\pi$. To see that
    it converges in this representation
    to a projection, one makes use of the resolvent
    equation
    \be
      Z_\nu(f) - Z_\mu(f) = (\mu - \nu) Z_\mu(f)Z_\nu(f)  \, .
    \ee 
    Sending first $\mu$ to infinity, where $\| Z_\mu(f) \|$ approaches
    $0$, one obtains the equality stated
    at the end of the proposition. Multiplying this equality with
    $\nu$ and proceeding again to the limit one finds that $(1-P(f)) P(f) = 0$.
    Clearly, $P(f)^* = P(f)$ since the approximating operators are
    selfadjoint.

    \medskip
    By construction, $P(f) \in \pi(\obfA)^{''}$. In order to see that it commutes
    with the elements of $\pi(\obfA)$, we make use of the fact, established
    in Lemma~\ref{l.2.4}, that
    $\blk \obfA \, Z_\mu(f) \brk = \blk Z_\mu(f) \, \obfA \brk$, $\mu > 0$.
    Since $(1 - P(f)) \, \pi(Z_\mu(f))=0$, it implies
    \be
    (1 - P(f)) \pi(\oA)  \pi(Z_\mu(f)) = 0 \, , \quad \oA \in \obfA \, , \
    \mu >  0 \, .
    \ee
    Multiplying this equality with $\mu$ and proceeding to the limit,
    one arrives at
    \be
    (1 - P(f)) \pi(\oA) P(f) = 0 \, , \quad  \oA \in \obfA \, . 
    \ee
    So the subspace $P(f) \cH$ is invariant under the action of $\pi(\obfA)$.
    Hence, $P(f)$ commutes with $\pi(\obfA)$ and is
    consequently an element of 
    $\pi(\obfA)^{''}\cap \, \pi(\obfA)^{'}$.
  \end{proof}
  It follows from this proposition that in factorial representations of
  $\obfA$, obtained from mixed phases by central decomposition,
  the projections $P(f)$ are either equal to $0$, or to $1$. In the former
  case, the resolvents $Z_\mu(f) = (\mu + a^*(f)a(f))^{-1})$, $\mu > 0$,
  are represented by $0$,
  \ie this observable disappears. In view of the elementary
  inequality
  \be
    \omega((\mu + a^*(f)a(f))^{-1}) \geq (\mu + \omega(a^*(f)a(f)))^{-1} \, ,
    \ee
    the single particle state $f \in L^2(\RR^s)$ is then infinitely
    occupied. Note, however, that the right hand side of this
  inequality may vanish in some state, whereas the left hand side is
  different from $0$.
  Only if the left hand side vanishes for some (hence all) $\mu$,
  the particles with wave function $f$ 
  form a proper condensate, cf.\ \cite{Bu4}.

  \medskip 
  The second possibility in a factorial representation, \ie $P(f) = 1$, implies
  that the resolvents $\pi(Z_\mu(f))$, $\mu > 0$, do not
  annihilate any vector in $\cH$, they are faithfully represented. The
  spectral decomposition of the corresponding number operator $a^*(f)a(f)$
  exists and determines the probabilities of finding a given number of
  particles in the single particle state $f$. These probabilities sum up to~$1$,
  so the particles do not form an infinite condensate in the state $\omega$.
  It shows that the resolvents are more sensitive to the presence of
  proper condensates than the one-particle density matrices.

  \medskip
  We conclude this section be presenting a family of non-interacting models
  that shed light on the formation of proper condensates and their
  interpretation. We consider single particle Hamiltonians $H_R$ with a
  continuous confining potential, such that the
  wave function $e_0$ of the ground state is constant in the bounded region
  $|\bix|^2 \leq R^2$ and decays rapidly at infinity.
  By specifying a sufficiently regular,
  positive function $e_0$ with these properties, such as
  \be
  e_0(\bix) = N
  \begin{cases}
    1 & \text{if} \quad |\bix|^2 \leq R^2 \\
    e^{- C \, (|\sbix|^2 - R^2)^3} & \text{if} \quad |\bix|^2 \geq R^2 \, , 
  \end{cases}  
  \ee
  one can determine corresponding polynomially bounded potentials,   
  making use of the stationary Schr\"odinger equation for $e_0$. The
  resulting Hamiltonian $H_R$ has a discrete spectrum with the unique
  normalized ground state $e_0$ and orthonormal excited states
  $e_k$, $k \in \NN$. Having fixed $H_R$, we consider corresponding
  $n$-particle states $\Omega_n \in \cF_n$ that are formed by the
  symmetric tensor product of $m(n)$ ground states $e_0$ and each of
  the excited states $e_k$, $1 \leq k \leq (n - m(n))$. The states on
  $\obfA$ that are determined by $\Omega_n$ are denoted by $\omega_n$,
  $n \in \NN$.
  Assuming that $m(n)$ tends to infinity in the limit of large $n$,
  the corresponding limit $\omega$ of this sequence of states returns
  on the basic resolvents the expectation values 
\be
\omega\big((\mu + a^*(f)a(f)^{-1} \big) = 
\begin{cases}
0  & \text{if} \quad \langle e_0 | f \rangle \neq 0 \\
\omega_\perp((\mu + a^*(f)a(f)^{-1}) &  \text{if} \quad
\langle e_0 | f \rangle = 0 \, .
\end{cases}
\ee
Here $\omega_\perp$ is the limit of the states formed by the 
tensor products of the excited states $e_k$, $k \in \NN$. This
result is obvious if one chooses $f = e_0$, respectively
$f = e_k$, $k \in \NN$. That it obtains for arbitrary $f \in \cS(\RR^s)$
requires some estimates, which are omitted here. 
We also put on record that
\be
\lim_{\mu \rightarrow \infty} \omega(\mu \, (\mu + a^*(f)a(f)^{-1}) = 1
\quad \text{if} \quad \langle e_0 | f \rangle = 0 \, .
\ee
Thus $\omega$ describes a proper condensate, regardless of how
fast $m(n)$ grows with $n$. Resolvents involving $f$ with a component
in the direction of $e_0$, no matter how small, disappear in the limit.
All resolvents with $f$ in the orthogonal complement of $e_0$ are
faithfully represented. They allow it to determine the occupation
numbers of the excited states. 

\medskip  
By construction, the $n$-particle vectors $\Omega_n$ are eigenstates of the
$n$-fold direct sum of $H_R$. The states $\omega_n$ are therefore stationary
with respect to the corresponding dynamics, $n \in \NN$, and this property
is inherited by their limit $\omega$. So in the GNS-representation induced
by $\omega$, the time translations are unitarily implemented. Since $\omega$
describes 
a confined system, the spatial transformations $\bix \in \RR^s$ 
cannot be unitarily implemented. They act on the basic
resolvents by shifting the support of the corresponding test functions,
$\bix \mapsto f_\sbix$. Picking any test function~$f$ with support in
the ball $|\bix|^2 \leq R^2$ that satisfies $\langle e_0 | f \rangle = 0$,
it follows from the choice of $e_0$ that
$\langle e_0 | f_\sbix \rangle = 0$ if the support of $f_\sbix$ also 
lies inside the ball. The corresponding resolvents are then
faithfully represented.
Yet if the support crosses the boundary, the function
$\bix \mapsto \langle e_0 | f_\sbix \rangle$ no longer vanishes and the
resolvents corresponding to such $f_\sbix$ disappear. So the following
picture emerges: The limit state $\omega$ describes a stationary system,
which is confined in the ball $|\bix|^2 \leq R^2$ with impenetrable walls. It
consists of the excitations $e_k$, $k \in \NN$. The proper condensate,
formed by the ground state $e_0$, cannot directly be observed. Its presence
manifests itself through the geometric constraints which it imposes on the
observables.  
Let us finally mention that by choosing for the number of ground
states $e_0$ in the approximating vectors $\Omega_n$ the multiplicities
$m(n) = n$, $n \in \NN$, the resulting limit state $\omega$ describes
the Fock vacuum inside the ball $|\bix|^2 \leq R^2$.
Thus, even in case of perfect condensation,
the limit states may exhibit a non-trivial, physically meaningful structure.

\section{Thermodynamic limit}
\label{sec3}
\setcounter{equation}{0} 

In this section, we study the thermodynamic limits of canonical
ensembles. There one changes the dynamics
by turning off the confining potentials. This causes problems on the
mathematical side since there may be no reasonable subalgebra
of $\obfA$ on which all intermediate dynamics act norm continuously. We
will discuss here how relevant information on the limit states can
be obtained in the algebraic setting, nevertheless. 

\medskip
The first, obvious question is: In what sense does the dynamics
$\balpha_L$ on $\obfA$, fixed by the Hamiltonian \eqref{e.1.2},
converge to the dynamics $\balpha$ without confining harmonic
potential if $L$ approaches infinity? 
An answer is given in \cite[App.]{Bu1}. Here we use 
the following, somewhat 
stronger version of this result, which is proven in an appendix.
\begin{lemma} \label{l.3.1}
  Let $M_n$ be an increasing sequence of positive numbers, tending to 
  infinity, and let $\varepsilon_n$ be a decreasing sequence, approaching
  $0$. There exists a sequence of lengths $L_n$ such that for any resolvent
  $R_\mu(f) = (\mu + a^*(f)a(f))^{-1}$ and any time $t$ such that 
  $2 \mu^{-3/2} \int_0^t \! ds \, \|(\biQ -s \biP)^2 f \| \leq M_n$, one has
  \be
  \|(\balpha_{L_n}(t) -
  \balpha(t))(R_\mu(f)\|_n \leq \varepsilon_n \, , \quad n \in \NN 
  \, .
  \ee 
 \end{lemma}
Given any operator $A_0$, which is a 
finite sum of products of basic resolvents and any finite time interval
for $t$, there exists some number $n_0$
such that the condition in this lemma is satisfied for the underlying 
resolvents if $n \geq n_0$. It is then straightforward to prove
that there exists some constant $C_{A_0}$, depending on~$A_0$, such that
\be
\|(\balpha_{L_n}(t) - \balpha(t))(A_0) \|_n \leq C_{A_0} \, \varepsilon_n \, ,
\quad n \geq n_0 \, .
\ee
Since the finite sums of products of resolvents are dense in $\bfA$
and $\rho_n(\bfA) = \rho_n(\obfA)$, 
it follows that 
\be
\lim_{n \rightarrow \infty} \| (\balpha_{L_n}(t) - \balpha(t))(\oA) \|_n =
0 \, , \quad \oA \in \obfA \, ,
\ee
uniformly for $t$ in compact subsets of $\RR$
(which we will hereafter abbreviate as “uniformly in $t$”).
The following corollary is an immediate consequence of this relation. 
\begin{corollary}
  Let $L_n$ be a sequence of lengths as in the preceding lemma
  and let~$\omega_n$ be stationary $n$-particle states for the dynamics
  $\balpha_{L_n}$, $n \in \NN$. Then all weak-*-limit points of this
  sequence are stationary states on $\obfA$ for the
  limit dynamics~$\balpha$. 
\end{corollary}  
\begin{proof}
  Since the states $\omega_n$ are invariant under the adjoint action of
  $\balpha_{L_n}(\RR)$, $n \in \NN$, one obtains for all $\oA \in \obfA$
  and $t \in \RR$
  \begin{align}
  \lim_{n \rightarrow \infty} |\omega_n(\oA  -  \balpha(t)(\oA)) | & = 
  \lim_{n \rightarrow \infty} |\omega_n(\balpha_{L_n}(t)(\oA) -
  \balpha(t)(\oA)) |
    \nonumber \\
  & \leq
  \lim_{n \rightarrow \infty} \| (\balpha_{L_n}(t) - \balpha(t))(\oA) \|_n
  = 0 \, .
  \end{align}
  Thus, all weak-*-limit points of the sequence of states are invariant
  under the adjoint action of $\balpha(\RR)$, as asserted. 
\end{proof}  
Restricting the limit states to the algebra $\bfC_{\sbalpha}$, there exists in
the corresponding GNS representations a continuous unitary representation of
the time translations. Furthermore, the limits of ground states, respectively
KMS states, are of the same kind. However, one must note that if the
lengths $L_n$ grow too rapidly with the particle number
$n$, the limit states
are the Fock vacuum. The GNS-representation of $\bfA$, which
is induced by the
Fock vacuum, is one-dimensional and thus trivially satisfies the KMS condition
for all temperatures. It is therefore  necessary to derive estimates
for the growth of $L_n$ for increasing $n$, which ensure that the particle
density of the states does not vanish in the limit of large particle numbers.

\medskip
We address this problem by focusing on the generators 
of the dynamics
$\balpha_L$, $\, L \in \RR_+ \cup \{ \infty \}$, whose action is
defined on a dense domain in $\bfC_{\sbalpha_L}$ by the derivations 
\be
\delta_L(C) \coloneqq
\tfrac{\text{\it \small d}}{\text{\it \small dt}} \, \balpha_L(t)(C) \,
\Big|_{t=0} \, .
\ee
The derivation determined by the homogeneous dynamics for $L = \infty$
is denoted by $\delta$. Since the domains 
depend on the dynamics and the choice of the
confinement length $L$, it is difficult to
compare the derivations. We therefore restrict them 
to the algebras $\rho_n(\bfA)$, $n \in \NN$, which 
allows us to treat these domain problems. 
Making use of the fact
that the representations $\rho_n$ are covariant, we put
\be
\delta_{L,n}(\rho_n(C)) \coloneqq 
\tfrac{\text{\it \small d}}{\text{\it \small dt}} \,
\text{Ad} \, e^{it \sbiH_{L,n}} (\rho_n(C)) \, \Big|_{t=0}
=  i \, [\biH_{L,n} , \, \rho_n(C)] \, . 
\ee
Whereas the dynamics $\balpha_{L,n}$ leave the algebras
$\rho_n(\bfA)$ invariant, $n \in \NN$, this is
not the case for the action of the
derivations. The root of this difficulty lies in the fact that
the interaction potentials in  $H_{L,n}$ have a continuous spectrum,
while the elements of $\bfA$ are composed of compact operators.

\medskip 
We cope with this problem by regularizing the
potentials.
Let $V \in C_0(\RR^s)$  be any two-body potential
and~$\biV_n$ the corresponding $n$-body potential,
which is obtained by summation over all pairs. The regularized 
potentials are defined in the strong operator topology by the
integrals 
\be
\biV_{\varepsilon, n} \coloneqq (1 / \varepsilon)
\int_0^\varepsilon \! ds \, \balpha_{0, n}(s)(\biV_n) \, , \quad
\varepsilon > 0 \, , 
\ee
where $\balpha_{0, n}$ is the non-interacting dynamics determined by
the Hamiltonian $\biH_{0,n}$, \ie the 
sum of momentum squares without confining potential. 
As has been shown in \cite[Lem.\ 4.1]{Bu1}, the regularized potentials  are
sums of tensor products of compact  and unit operators, 
which define bounded derivations on $\rho_n(\bfA)$, $n \in \NN$. They commute
with the spatial translations, so the
resulting interactions are still homogeneous. 
The regularization has only a minor effect on the time evolution,
as the  following lemma shows. Its proof is given in
an appendix.
\begin{lemma} \label{l.3.3} 
  Let $n \in \NN$, let $\biH_n \coloneqq \biH_{0,n} + \biV_n$ be the
  homogeneous 
  $n$-particle Hamiltonian for the given two-body potential,
  and let
  $\biH_{\varepsilon, n} \coloneqq \biH_{0,n} + \biV_{\varepsilon, n}$
  be the regularized Hamiltonian. If $\varepsilon$ approaches $0$,
  the unitaries
  $e^{it \sbiH_{\varepsilon , n}}$ converge in
  norm to $e^{it \sbiH_{n}}$, uniformly in $t$.
\end{lemma}
Given these results, it is justified to restrict attention to the
regularized interaction potentials when analyzing the derivations. In
the following lemma, the domain problems are then settled.

\begin{lemma} \label{l.3.4} 
  Let $n \in \NN$ and let $\delta_{L,n}$ be the derivations
  acting on $\rho_n(\bfA)$ on their respective domains, with  
  $L \in \RR_+ \cup \{  \infty \}$ and regularized 
  interaction potential $\biV_{\varepsilon,n}$, $n \in \NN$.
  Furthermore, let $\bfA_0$
  be the algebra consisting of finite sums of products of the basic
  resolvents. All derivations $\delta_{L,n}$ are defined on $\rho_n(\bfA_0)$
  with range in $\rho_n(\bfA)$. In fact, $\rho_n(\bfA_0)$ is a core
  (essential domain) for each of the derivations. 
\end{lemma}
\begin{proof}
  Since $\| \biV_{\varepsilon, n} \|_n \leq n(n-1) \|V\|$,
  the operators $i \biV_{\varepsilon, n}$ act as bounded derivations
  on $\rho_n(\bfA)$. Hence the domain of
  $\delta_{L,n}$ coincides with the domain of
  $\delta_{0,L,n}$, \ie the derivation determined
  by the Hamiltonian $\biH_{0,L,n}$ 
  with vanishing  interaction potential. Given a resolvent
  $R_\mu(f) = (\mu + a^*(f)a(f))^{-1}$, $f \in \cS(\RR^s)$, one has 
  \begin{align}
     [\biH_{0,L}, R_\mu(f)] & = R_\mu(f) \, [a^*(f)a(f), \biH_{0,L}] \, R_\mu(f)
    \nonumber \\
    &  = R_\mu(f) \, \big( a^*(f)a(H_{0,L,1}f)
    - a^*(H_{0,L,1} f) a(f) \big) \, R_\mu(f) \, ,
  \end{align}
  where $H_{0,L,1} = (\biP^2 + L^{-4} \, \biQ^2)$ is  the Hamiltonian on the
  one-particle space. Hence 
  \mbox{$H_{0,L.1} f \in \cS(\RR^s)$}. If one evaluates this relation in the
  representation $\rho_n$, one obtains
  \begin{align}
  &  \delta_{0,L,n}(\rho_n(R_\mu(f)) = [iH_{0,L,n}, \rho_n(R_\mu(f))] \nonumber \\
    &    = \rho_n\big(R_\mu(f) \, i\big(a^*(f)a(H_{0,L,1}f) -
    a^*(H_{0,L,1}f)a(f)\big) \,   R_\mu(f) \big) \, .
  \end{align} 
The middle term $i \rho_n(a^*(f)a(H_{0,L,1}f) - a^*(H_{0,L.1}f)a(f))$  
is a finite direct sum of rank-one operators. This  implies that the
right hand side of the preceding equalities is an element of $\rho_n(\bfA)$,
cf.\ \cite[Lem.\ 3.3]{Bu1}. Since $\delta_{0,L,n}$ is a
derivation on $\rho_n(\bfA_0)$  which satisfies the Jacobi identity,
it follows that its
range is contained in $\rho_n(\bfA)$. It shows that $\rho_n(\bfA_0)$  is a
common domain of the derivations  $\delta_{0,L,n}$ for all
\mbox{$L \in \RR_+ \cup \{ \infty \}$} and $n \in \NN$.

\medskip
For the proof that $\rho_n(\bfA_0)$ is a core for $\delta_{0,L,n}$, where 
$L$ is now kept fixed, we use the fact that for any $f \in \cS(\RR^s)$ one has
$e^{itH_{0,L,1}}f \in \cS(\RR^s)$ for $t \in \RR$. Hence
\be
t \mapsto \rho_n(\alpha_{0,L}(t)(A_0)) = \text{Ad} \, e^{itH_{0,L,n}}
    (\rho_n(A_0)) \in \rho_n(\bfA_0) \, , \quad A_0 \in \bfA_0 \, .
\ee
These functions are continuous with respect to the norm on $\cF_n$.
Given \mbox{$A_0 \in \bfA_0$}, the regularized operators
$\alpha_{0,L}(\varphi)(A_0) \in \bfC_{\sbalpha_{0,L}}$, defined in 
equation~\eqref{e.2.1}, can be approximated in the
representation $\rho_n$ by norm convergent Riemann sums whose values
lie in $\rho_n(\bfA_0)$. The derivation $\delta_{0,L,n}$ 
commutes with the action of
$\text{Ad} \, e^{itH_{0,L,n}}$, $t \in \RR$, and maps
$\rho_n(A_0)$ into $\delta_{0,L,n}(\rho_n(A_0)) \in \rho_n(\bfA)$.
For the latter operator,
the resulting Riemann sums are norm convergent as well. Hence
$\rho_n(\alpha_{0,L}(\varphi)(A_0))$ and
$\rho_n(\delta_{0,L}(\alpha_{0,L}(\varphi)(A_0)))$ can be approximated by
sequences in $\rho_n(\bfA_0)$ with respect to the graph-norm
\be
A_0 \mapsto \| \rho_n(A_0) \|_n + \|\delta_{0,L,n}(\rho_n(A_0))\|_n \, .
\ee
Thus, $\rho_n(\bfA_0)$ is a core for $\delta_{0,L,n}$, $n \in \NN$,
and all $L \in \RR_+ \cup \{ \infty \}$, which completes the proof. 
\end{proof}  

The derivations $\delta_L$, $L \in \RR_+ \cup \{ \infty \}$,
are well-behaved on the algebra
$\bfA_0$ in the representations $\rho_n$, $n \in \NN$, but they are 
unbounded on $\bfA_0$. Here  we can make use of Lemma \eqref{l.2.2} 
according to which the states in the GNS-representation of an
$\balpha_L$-invariant state $\omega$ on the algebra $\bfC_{\sbalpha_L}$ 
can be continuously
extended to the dynamical algebra $\bfA_{\sbalpha_L}$,  hence in
particular to the functions $t \mapsto \alpha_L(t)(A_0)$, $A_0 \in \bfA_0$.
Now if $C_1, C_2 \in \bfC_{\sbalpha_L}$ are regularized operators, one has
\be
t \mapsto \omega(C_1 \, \alpha_L(t)(A_0) \, C_2) =
\omega(\alpha_L(-t)(C_1) \, A_0 \, \alpha_L(-t)(C_2)) \, ,
 \ee
 which is arbitrarily often differentiable. So the functions
 $t \mapsto \alpha_L(t)(A_0)$ are defined in these
 representations as smooth sesquilinear forms. In particular, the functions 
 $t \mapsto \omega(\alpha_L(t)(A_0))$, $A_0 \in \bfA_0$, are constant, so their
 derivative vanishes. This leads us to propose the relation 
\be 
\lim_{n \rightarrow \infty} \, \omega_n(\delta(A_0)) = 
\lim_{n \rightarrow \infty} \,
\tfrac{{d}}{{dt}} \, \omega_n(\balpha(t)(A_0)) \,  
\Big|_{t=0}  =  
 0
 \, , \quad A_0 \in \bfA_0 \, .
 \ee
 as a general test for determining whether a sequence of
 $n$-particle states $\omega_n$, $n \in \NN$, has limit points that
 are invariant with respect to the adjoint action of
 the homogeneous dynamics $\balpha$.
 In the following theorem, we apply this test to
 a sequence of states that are stationary with respect to
 the confined dynamics~~$\balpha_L$. The assumption that it approximates a
 stationary state for the homogeneous dynamics $\balpha$ leads to 
 a condition on the confinement length as a function of the
 number of particles.
  \begin{theorem}
   Let $\balpha_L$ be dynamics with a fixed regularized two-body interaction
   potential and arbitrary confinement length $L$, and   
   let $\omega_{L,n}$ be $n$-particle states that are stationary with
   respect to the adjoint action of $\balpha_L$, $L \in \RR_+$.
   For any sequence $n \mapsto L_n$ satisfying 
   $n/L_n^8 \rightarrow 0$ in the limit of large $n$, one has
   \be
   \lim_{n \rightarrow \infty} \omega_{L_n,n}(\delta(A_0)) = 0 \, ,
   \quad A_0 \in \bfA_0 \, .
   \ee  
 \end{theorem}  
 \begin{proof}
   The derivations $\delta_L$ and $\delta$ differ by the harmonic potential, 
    scaled with the factor $L^{-4}$. A computation as in the proof of
   Lemma \ref{l.3.4} yields for any basic resolvent
   $R_\mu(f) = (\mu + a^*(f)a(f))^{-1}$ the equality
   \be
   (\delta - \delta_L)(R_\mu(f)) =
   L^{-4} \, R_\mu(f) \big (a^*(\biQ^2 f) a(f) - a^*(f) a(\biQ^2 f) \big)
   R_\mu(f) \, .
   \ee
   The operators $R_\mu(f)$, $a(f) R_\mu(f)$, and $R_\mu(f) a^*(f)$
   are bounded, but the norms of $a(\biQ^2 f)$ and of its adjoint increase
   on the $n$-particle spaces $\cF_n$ with $n \in \NN$. There one obtains
   \be
   \| a(\biQ^2 f) \|_n = \| a^*(\biQ^2 f) \|_n \leq n^{1/2} \| \biQ^2 f \|
   \, .
   \ee
   So one arrives at the bound 
   \be
   \| (\delta - \delta_L)(R_\mu(f)) \|_n \leq 2 \mu^{-3/2} \| \biQ^2 f \|
   \, (n^{1/2}/L^4) \, .
   \ee
   Since the elements of the algebra $\bfA_0$ are finite sums of products
   of the basic resolvents and $(\delta - \delta_L)$ is a derivation,
   which satisfies the Jacobi identity,
   it follows that for each $A_0 \in \bfA_0$
   there is some constant $C_{A_0}$ such that
   \be
   \| (\delta - \delta_L)(A_0) \|_n \leq C_{A_0} (n^{1/2}/L^4) \, ,
   \quad n \in \NN \, .
   \ee
   For sequences $n \mapsto L_n$ as in the statement and corresponding
   stationary states $\omega_{L_n,n}$, $n \in \NN$, one has that
   $\omega_{L_n,n}(\delta(A_0)) = \omega_{L_n,n}((\delta - \delta_{L_n})(A_0))$,
   $A_0 \in \bfA_0$. But
   \be
   | \omega_{L_n,n}((\delta - \delta_{L_n})(A_0)) | \leq
   \| (\delta - \delta_{L_n})(A_0) \|_n \leq C_{A_0} (n^{1/2}/L_n^4) \, ,
   \ee
   so the statement follows. 
 \end{proof}
 The condition $n/L_n^8 \rightarrow 0$,
 established in this theorem, implies that for attractive as well as
 repulsive two-body potentials the limit states lie in the kernel of the
 derivation $\delta$ and are stationary in this sense. However, it does
 not impose constraints on the density of the
 limit states, which may vanish, be finite, or diverge.  
 The condition is compatible with relations between the particle
 number and the confinement length, which are based
 on constraints on the density profile of the limit states,
 cf.\ for example \cite{DaGiPiSt, LiSeSoYn}. 
 The differences between those conditions and the present one allow  
 for limit states that are stationary, but have different profiles. 

\section{Conclusions}
\label{sec4}
\setcounter{equation}{0} 

In this paper, we have demonstrated the usefulness of the resolvent algebra
for the analysis and 
interpretation of bosonic systems in the limit of an infinite number
of particles, respectively of infinite volume.
We have considered arbitrary two-body potentials, whether respulsive
or attractive and whether short- or long range. 
With little additional effort, the assumption that the potentials
are continuous can be relaxed \cite[Sec.\ 6]{Bu3}. 
As we have shown, the resolvent algebra
provides a solid basis for analyzing the limit states.
The transition from finite to infinite systems is
possible without changing the framework. In all these cases, the states
are defined on the projective limit of the resolvent algebra, which is
completely determined by the family of representations with 
a finite number of particles. Moreover, 
the projective limit carries the action of the underlying dynamics
and contains the regularized dynamical algebras on which the automorphisms
act pointwise norm-continuously. This fact considerably simplifies the
proof that all limits of confined
Gibbs ensembles are either equilibrium states,
satisfying the KMS condition, or they are ground states.

\medskip
The basic resolvents of the algebra are key tools for
analyzing the phase structure of limit states, such as the
appearance of Bose-Einstein 
condensates. As a primary ingredient in the central
decomposition of mixed states into pure phases, they enable
a refined analysis that signals the formation of proper condensates
in confined systems 
and sheds new light on their physical interpretation.

\medskip
The thermodynamic limit can be investigated within this 
framework as well. There the
effects of changes of the dynamics caused by
the turning off of the confinement forces has to be
controlled. For finite systems,
the pointwise norm convergence of the confined dynamics
to the spatially homogeneous dynamics has been demonstrated. 
However, the quantitative estimates of the
convergence rates for large particle numbers
are not yet sufficient to control the
properties of the limit states, such as the particle density.
Here, the transition from the automorphic action of the dynamics to
the corresponding 
generators (derivations) allowed for a better control of these rates.
Again, the basic resolvents proved useful, as they 
form an essential domain (core) for the derivations with
arbitrary confinement length and arbitrary 
regularized two-body interactions. Using this
information as starting point, we were able to
establish a relation between the number of particles in the
system and the confinement length that ensures the
stationarity of the thermodynamic
limit of confined states. This relationship encompasses
the results in the literature that are based on controlling
the particle density of the limiting states.

\medskip
So we conclude that the resolvent algebra 
is not only useful for analyzing specific models, 
as has been demonstrated   
for lattice systems in \cite{BuGr,DeLaLe}.  
The present results provide evidence that the resolvent algebra  
is also valuable for the study  of continuous bosonic systems.
This includes questions that go beyond the investigation of condensates,
such as the study of the breaking of geometric symmetries,
for example in crystals. 
We hope that the present results will inspire further research
on the algebraic approach to many-body physics,
which complements the large body of work on this topic.


\appendix

\section{Proof of Lemma \ref{l.3.1}}
\label{app1}
\setcounter{equation}{0} 

The proof of this lemma is given in two steps. First, we show that the
Dyson operators $e^{itH_{L,n}} e^{-itH_{0,L,n}}$, where
$H_{0,L,n}$ denotes the Hamiltonian without interaction
potential, converges for large $L$ in norm on $\cF_n$
to $e^{itH_{n}} e^{-itH_{0,n}}$, $n \in \NN$.
There it is essential that the interaction potential is
bounded. Then
we establish the norm convergence of the non-interacting dynamics
$\text{Ad} \, e^{itH_{0,L,n}}$ on the basic resolvents to
$\text{Ad} \, e^{itH_{0,n}}$.
Since the Hamiltonians $H_{0,L,n}$ and $H_{0,n}$ differ by some unbounded
operator, the properties of the respective resolvents then matter. 

\medskip
The terms in the series expansion of the Dyson operators are 
integrals of \mbox{$s \mapsto \text{Ad} \, e^{isH_{0,L,n}}(\biV_n)$},
where the action of the dynamics on
the two-body potentials in $\biV_n = \sum_{k \neq l, 1}^n V(\biQ_k - \biQ_l)$
is given~by  
\begin{align}
& s \mapsto \text{Ad} \, e^{isH_{0,L,n}}(V(\biQ_k - \biQ_l)) \nonumber \\
& = V(\cos(2s/L^2)(\biQ_k - \biQ_l) + L^2 \sin(2s/L^2) (\biP_k - \biP_l)) \, .
\end{align}
The operators $\biQ \coloneqq 2^{-1/2}(\biQ_k - \biQ_l)$,
$\biP \coloneqq 2^{-1/2}(\biP_k - \biP_l)$ are canonically conjugate. 
Absorbing a factor $2^{1/2}$ in the potential $V$, the first order term
in the Dyson expansion is given by a sum of $n (n-1)$ operators of the form
$ t \mapsto \int_0^t \! ds \, \text{Ad} \, e^{itH_{0,L,1}}(V(\biQ)) $, 
where $H_{0,L,1} = (\biP^2 + L^{-4} \, \biQ^2)$. For the proof that this
integral converges for large $L$ in norm to the integral
where $H_{0,L,1}$ is replaced by $H_{0,1} = \biP^2$, we consider first 
potentials $V \in \cS(\RR^s)$ and introduce some time-dependent intermediary
dynamics $s \mapsto \text{Ad} \, (e^{\, i(f_L(s) \, H_{0,1}})$, where the
function $f_L$ is to be determined. Now
\begin{align}
&  \| \text{Ad} \, e^{isH_{0,L,1}}(V(\biQ)) -
  \text{Ad} \, e^{if_L(s)H_{0,1}}(V(\biQ)) \| \nonumber \\
& = \| \text{Ad} \, e^{-isH_{0,L,1}} e^{if_L(s)H_{0,1}}(V(\biQ)) - V(\biQ) \| 
  \nonumber \\
& = \| \int_0^s \! dr \,  {\tfrac{d}{dr}}
   \text{Ad} \, e^{-irH_{0,L,1}} e^{if_L(r)H_{0,1}}(V(\biQ)) \| \nonumber \\
& =   \| \int_0^s \! dr \,  
   \text{Ad} \, e^{-irH_{0,L,1}} e^{if_L(r)H_{0,1}} \big(
   \big[\text{Ad} \, e^{-if_L(r)H_{0,1}}(f_L(r)'H_{0,1} - H_{0,L,1}), V(\biQ)
     \big] \| \big) \nonumber \\  
  & \leq  \int_0^s \! dr \, \| \big[
    (f_L'(r) \biP^2 - (\biP^2 + L^{-4}(\biQ - f_L(r) \biP)^2), V(\biQ) \big]
  \| \, . 
\end{align}  
If the terms proportional to $\biP^2$ in the last
line are to vanish, the function $f_L$
must satisfy the differential equation $f_L' = 1 + L^{-4} f_L^2$. 
Choosing the solution 
$r \mapsto f_L(r) = L^2 \tan(r/L^2)$ for  $0 \leq r < \pi L^2/ 2$,
we obtain for small $s/L^2$ 
\begin{align}
&  \| \text{Ad} \, e^{isH_{0,L,1}}(V(\biQ)) -
   \text{Ad} \, e^{if_L(s)H_{0,1}}(V(\biQ)) \| \nonumber \\
& \leq \int_0^s \! dr \, \| \big[
     (\tan(r/L^2)/L^2)(\biP \biQ + \biQ \biP), V(\biQ) \big] \| \nonumber \\
   &  \leq  \int_0^{s/L^2} \! dr \, \tan(r) \sup_{\bix} | \bix \bnabla V(\bix) |
   \leq c s^2/L^4 \, ,
\end{align}
where here and in the following
$c$ is some constant, which only depends on the potential. Hence we get  
\be
\| \int_0^t \! ds \, \big( \text{Ad} \, e^{isH_{0,L,1}}(V(\biQ)) -
   \text{Ad} \, e^{if_L(s)H_{0,1}}(V(\biQ)) \big) \| 
\leq c \, t^3 / L^4 \, . 
\ee
Next, we consider the passage from the action of the 
intermediate dynamics to that of the unconfined one. There we get
for small $t/L^2$ by a change of variable 
\begin{align}
  & \| \int_0^t \! ds \,
  \big(\text{Ad} \, e^{if_L(s) H_{0,1}}(V(\biQ)) -
  \text{Ad} \, e^{isH_{0,1}}(V(\biQ)) \big)\|  \nonumber \\ 
  & = \| \int_0^{L^2 \tan(t/L^2)} \! ds \,
  \tfrac{1}{1 + (s^2/L^4)} \, \text{Ad} \, e^{isH_{0,1}} (V(\biQ))
  - \int_0^t \! ds \, \text{Ad} \, e^{isH_{0,1}} (V(\biQ)) \| \nonumber \\
  & \leq \| \int_t^{L^2 \tan(t/L^2)} \hspace{-15pt} 
  ds \, \text{Ad} \, e^{isH_{0,1}} (V(\biQ)) \|
  + \| \int_0^{L^2 \tan(t/L^2)} \hspace{-15pt}   ds \, 
  \tfrac{(s^2/L^4)}{1 + (s^2/L^4)} \, \text{Ad} \, e^{isH_{0,1}}
  (V(\biQ)) \| \nonumber \\[1mm]
  & \leq \big(|L^2 \tan(t/L^2) - t| + L^2 \tan^3(t/L^2) \big)
  \sup_{\bix} | V(\bix) | \leq c \, t^3/L^4  \, .   
\end{align}
Combining these estimates, we obtain 
\be
\| \int_0^t \! ds \, \big(
\text{Ad} \, e^{isH_{0,L,1}}(V(\biQ)) -
 \text{Ad} \, e^{isH_{0,1}}(V(\biQ)) \big) \|
  \leq c \, t^3/L^4 \, . 
\ee 
Since the test functions $\cS(\RR^s)$ are dense in $C_0(\RR^s)$ with regard to
the supremum norm, we can proceed now to arbitrary
potentials $V \in C_0(\RR^s)$. The preceding estimate then implies
\be
\lim_{L \rightarrow \infty} \| \int_0^t \! ds \, \big(
\text{Ad} \, e^{isH_{0,L,1}}(V(\biQ)) -
 \text{Ad} \, e^{isH_{0,1}}(V(\biQ)) \big) \| = 0 \, ,
\ee
uniformly in $t$. Replacing
$V$ by $\biV_n$ and the one-particle dynamics by
$H_{0,L,n}$, respectively $H_{0,n}$, it proves the
norm convergence of the first order term in the Dyson expansion in
the limit of large confinement lengths $L$ and any number of particles
$n \in \NN$. 

\medskip 
The convergence of the whole series follows then by induction,
making use of the fact that it is absolutely convergent for given 
$n \in \NN$ and compact sets in time~$t$. We briefly sketch the argument.
Let
\begin{align}
&  \bSigma_{L,n,k}(t)  \coloneqq \\
& 
\int_0^t \!\!\!
  ds_k \, \text{Ad} \, e^{is_k H_{0,L,n}}(\biV_n) \int_0^{s_{k}} \!\! \!
  ds_{k-1} \, \text{Ad} \, e^{is_{k-1}H_{0,L,n}}(\biV_n) \cdots
  \int_0^{s_{2}} \!\!\!
  ds_1 \, \text{Ad} \, e^{is_1 H_{0,L,n}}(\biV_n) \nonumber 
\end{align}
be the $k$th term in the Dyson expansion of $e^{itH_{L,n}} e^{-itH_{0,L,n}}$ 
and, similarly, $\bSigma_{n,k}(t)$ the $k$th term in the expansion 
of $e^{itH_{n}} e^{-itH_{0,n}}$, $k \in \NN$.
Assuming that  
\be
\lim_{L \rightarrow \infty} \| \bSigma_{L,n,k}(t) - \bSigma_{n,k}(t) \|_n = 0 \, ,
\ee 
uniformly in $t$, the step to $k + 1$ is accomplished by
rearranging the difference and by a partial integration. One has
\begin{align}
& \bSigma_{L,n,k+1}(t) - \bSigma_{n,k+1}(t) \nonumber \\
  & =   \int_0^t \! ds \,  \big( \tfrac{d}{ds} 
  \int_0^s \! dr \, \big( \text{Ad} \,
  e^{irH_{0,L,n}}(\biV_n) - \text{Ad} \,
  e^{irH_{0,n}}(\biV_n) \big) \big) \bSigma_{L,n,k}(s) \nonumber \\
& + \int_0^t \! ds \, \text{Ad} \,
  e^{isH_{0,n}}(\biV_n) \, \big(\bSigma_{L,n,k}(s) - \bSigma_{n,k}(s)
  \big) \nonumber \\ 
& =   
  \int_0^t \! dr \, \big( \text{Ad} \,
  e^{irH_{0,L,n}}(\biV_n) - \text{Ad} \,
  e^{irH_{0,n}}(\biV_n) \big) \, \bSigma_{L,n,k}(t) \nonumber \\
& -   \int_0^t \! ds \,  \big(
  \int_0^s \! dr \, \big( \text{Ad} \,
  e^{irH_{0,L,n}}(\biV_n) - \text{Ad} \,
  e^{irH_{0,n}}(\biV_n) \big) \big)
  \tfrac{d}{ds}  \bSigma_{L,n,k}(s) \nonumber \\
& + \int_0^t \! ds \, \text{Ad} \,
  e^{isH_{0,n}}(\biV_n) \, \big(\bSigma_{L,n,k}(s) - \bSigma_{n,k}(s)
  \big) \, .  
\end{align}  
The first two terms on the right hand side of the second equality contain
the difference of the first order contributions, which vanishes for
large $L$, and the last term contains the difference of the terms in the
induction hypothesis. Since
$\| \bSigma_{L,n,k}(s) \|_n \leq s^k \| \biV_n \|^k / k!$
and
$\| \frac{d}{ds} \bSigma_{L,n,k}(s) \|_n \leq s^{k-1} \| \biV_n \|^{k} /
(k-1)! $, 
one obtains
\begin{align} 
& \|\bSigma_{L,n,k+1}(t) - \bSigma_{n,k+1}(t) \|_n \nonumber \\
  & \leq 2 \, t^k \, \| \biV_n \|_n^k / k!
   \sup_{0 \leq s \leq t} \| \int_0^s \! dr \, \big( \text{Ad} \,
  e^{irH_{0,L,n}}(\biV_n) - \text{Ad} \,
  e^{irH_{0,n}}(\biV_n) \big) \|_n  \nonumber \\
  & + t \, \| \biV_n \|_n \sup_{0 \leq s \leq t}
  \|\bSigma_{L,n,k}(s) - \bSigma_{n,k}(s)\| \, .
\end{align}   
It shows that also the $(k+1)$st term converges in norm for large $L$,
uniformly in $t$, completing the induction. In view of the
absolute norm convergence of the Dyson expansion, we have thus established
\be \label{e.a.12} 
\lim_{L \rightarrow \infty} \| e^{itH_{L,n}} e^{-itH_{0,L,n}} -
e^{itH_{n}} e^{-itH_{0,n}} \|_n = 0 \, , 
\ee
uniformly in $t$.

\medskip
For the proof that  the non-interacting dynamics converges for large $L$ on the
basic resolvents, we proceed from the equalities 
\begin{align}
& \| \text{Ad} \, e^{itH_{0,L,n}}(R_\mu(f)) -
  \text{Ad} \, e^{itH_{0,n}}(R_\mu(f)) \|_n   \\ 
&  = \| R_\mu(e^{itH_{0,L,1}} \! f) - R_\mu(e^{itH_{0,1}} \! f) \|_n \nonumber \\
&  = \|R_\mu(e^{itH_{0,L,1}} \! f) \big( a^*(e^{itH_{0,1}} f)a(e^{itH_{0,1}} \! f)
  -  a^*(e^{itH_{0,L,1}} \! f ) a(e^{itH_{0,L,1}} \!f ) \big)
  R_\mu(e^{itH_{0,1}} \! f) \|_n \, .
\nonumber 
\end{align}
Now $\| R_\mu(g) \|_n \leq \mu^{-1}$ and 
$\| a(g) R_\mu(g) \|_n = \| R_\mu(g) a^*(g) \|_n
\leq \mu^{-1/2}$ for \mbox{$g \in \cS(\RR^s)$}.
Hence, we obtain for resolvents satisfying the condition in the lemma 
\begin{align}
& \| \text{Ad} \, e^{itH_{0,L,n}}(R_\mu(f)) -
  \text{Ad} \, e^{itH_{0,n}}(R_\mu(f)) \|_n  \nonumber \\[1mm]  
&  \leq 2 \mu^{-3/2} \, \| a(e^{itH_{0,1}}f ) - a(e^{itH_{0,L,1}}f ) \|_n
  \leq  2 \mu^{-3/2} \, \|( e^{itH_{0,1}} - e^{itH_{0,L,1}}) f \| \, n^{1/2}
\nonumber \\ 
&  =  2 \mu^{-3/2} \, \| \int_0^t \! ds \, e^{-isH_{0,L,1}} e^{isH_{0,1}} \, 
L^{-4} (\biQ - s \biP)^2 f \, \| \, n^{1/2}  \nonumber \\[-1mm] 
& \leq L^{-4} \, 2 \mu^{-3/2} \int_0^t \! ds \,
\| (\biQ - s \biP)^2 f \| \, n^{1/2} \ \leq \ L^{-4} M_n \, n^{1/2} \, .
\end{align} 
By combining these estimates and using the fact that the action of 
$\balpha_L(t) \balpha_{0,L}^{-1}(t)$ and $\balpha(t) \balpha_{0}^{-1}(t)$
on observables in $\cF_n$
is described by the adjoint action of the respective Dyson operators, we obtain
\begin{align}
&  \| \balpha_L(t)(R_\mu(f)) - \balpha(t)(R_\mu(f)) \|_n \nonumber \\
&  = \|
  \balpha_L(t) \balpha_{0,L}^{-1}(t) \big( \balpha_{0,L}(t)(R_\mu(f)) \big)
  - \balpha(t) \balpha_{0}^{-1}(t) \big( \balpha_{0}(t)(R_\mu(f))
  \big)\|_n \nonumber \\
&  \leq \| \balpha_{0,L}(t)(R_\mu(f) - \balpha_{0}(t)(R_\mu(f) \|_n
  \nonumber \\
& + \| \big( \balpha_L(t) \balpha_{0,L}^{-1}(t) - \balpha(t) \balpha_{0}^{-1}(t)
  \big)(\balpha_0(t)(R_\mu(f))) \|_n  \nonumber \\
  & \leq L^{-4} \, M_n \, n^{1/2} + 2 \mu^{-1/2} \,
  \|e^{itH_{L,n}} e^{-itH_{0,L,n}} -   e^{itH_{n}} e^{-itH_{0,n}} \|_n \, .  
\end{align}  
We can choose now for the confinement length $L$ a length
$L_n$ such that this upper bound
is equal to the given $\varepsilon_n$. It completes the proof of the
lemma. 


\section{Proof of Lemma \ref{l.3.3}}
\label{app2}
\setcounter{equation}{0} 

For the proof we can make use of the Dyson operators, noting that 
the unitaries $\varepsilon \rightarrow e^{it \sbiH_{\varepsilon , n}}$
  converge for small $\varepsilon$ 
  in norm to $e^{it \sbiH_{n}}$ if and only if 
  $\varepsilon \rightarrow e^{it \sbiH_{\varepsilon , n}} e^{-it \sbiH_{0,n}}$
  converges in norm to $ e^{it \sbiH_{n}} e^{-it \sbiH{0,_n}} $.
  Moreover, since $\biV_{\varepsilon,n}$
  is bounded, we can proceed to the Dyson expansion
    $e^{it \sbiH_{\varepsilon , n}} e^{-it \sbiH_{0,n}} =
  1 + \sum_{k=1}^\infty i^k \, \bSigma_{\varepsilon, n, k}(t) $, where  the
  summands are given by 
    \begin{align}
& \bSigma_{\varepsilon, n, k}(t) \\ 
& \coloneqq
\int_0^t \! \! \!  ds_k \,  \ad{s_k}(\biV_{\varepsilon,n}) 
\int_0^{s_k} \! ds_{k-1} \cdots
\int_0^{s_2} \! \!  ds_1 \, \ad{s_1}(\biV_{\varepsilon,n}) \, . \nonumber     
    \end{align}
    The summands containing the unregularized potential $\biV_n$
    are denoted by $\bSigma_{0,n.k}$. 
The Dyson series converge absolutely in norm and uniformly with regard to
$\varepsilon > 0$ and $t$. So it suffices to establish
convergence for the individual terms in these sums.
The proof of convergence of $\bSigma_{\varepsilon,n,k}(t)$
for small $\varepsilon$ is given by induction.
For $k=1$ one has 
\begin{align}
  & \bSigma_{\varepsilon,n,1}(t) = \int_0^{t} \! ds \,
  \ad{s}(\biV_{\varepsilon,n}) \nonumber \\
& = (1/\varepsilon) \int_0^{\varepsilon} \! dr 
\int_0^{t} \! ds  \, \ad{(r+s)}(\biV_n) 
= (1/\varepsilon) \int_0^{\varepsilon} \! dr 
\int_r^{t+r} \! ds  \, \ad{s}(\biV_n) \nonumber \\
& = \bSigma_{0,n,1}(t) + (1/\varepsilon) \int_0^\varepsilon \! dr \,
\big( \bSigma_{0,n,1}(t+r) - \bSigma_{0,n,1}(r) - \bSigma_{0,n,1}(t) \big) \, .
\end{align}
Now
\begin{align}
  & \| \bSigma_{0,n,1}(t+r) - \bSigma_{0,n,1}(t) - \bSigma_{0,n,1}(r) \|
  \nonumber \\ 
& \leq \int_t^{t+r} \! ds \, \| \biV_n \|_n + \int_0^r \! ds \, \|\biV_n\|_n
= 2r \, \|\biV_n\|_n \, ,
\end{align}
hence
\begin{align}
& (1/\varepsilon) \| \int_0^\varepsilon \! dr \, \big(
  \bSigma_{0,n,1}(t+r) - \bSigma_{0,n,1}(t) - \bSigma_{0,n,1}(r) \big) \|_n
  \nonumber \\ 
  & \leq (2/\varepsilon) \int_0^\varepsilon \! dr \, r \, \| \biV_n \|_n =
\varepsilon \, \|\biV_n\|_n \, .
\end{align}
So $\varepsilon \mapsto \bSigma_{\varepsilon,n,1}(t)$ converges
in norm to $\bSigma_{0,n,1}(t)$, uniformly in $t$. 
Assuming that $\bSigma_{\varepsilon,n,k}(t)$ converges in this manner to
$\bSigma_{0,n,k}(t)$, the step to $k+1$ is established as follows. One has
\begin{align}
&  \bSigma_{\varepsilon, n, k+1}(t) = \int_0^t \! ds \, \ad{s}(\biV_n) \, 
  \bSigma_{0,n,k}(s) \nonumber \\
& + \int_0^t \! ds \, \ad{s}(\biV_{\varepsilon,n} - \biV_n) \, \bSigma_{0,n,k}(s)
  \nonumber \\   
&  +
  \int_0^t \! ds \, \ad{s}(\biV_{\varepsilon,n}) \, 
  \big(\bSigma_{\varepsilon,n,k}(s) -  \bSigma_{0,n,k}(s)\big) \, .
\end{align}
The integral on the right hand side of the equality sign  is equal to
$\bSigma_{0,k+1}(t)$. 
So it must be shown that the terms in the second and third
line vanish in norm for small 
$\varepsilon$, uniformly in $t$. 
The norm of the integral in the last line can be estimated by
\[
\int_0^t \! ds \, \| \biV_n \|_n \,
\| \bSigma_{\varepsilon,n,k}(s) -  \bSigma_{0,n,k}(s) \|_n \, .
\]
Making use of the induction hypothesis, it tends to $0$ for 
small $\varepsilon$, uniformly in $t$. For the 
integral in the second line one obtains by partial integration 
\begin{align}
  & \int_0^t \! ds \, \ad{s}(\biV_{\varepsilon,n} - \biV_n)
  \bSigma_{0,n,k}(s) \nonumber \\
&  = 
  \int_0^t \! dr \, \ad{r}(\biV_{\varepsilon,n} - \biV_n) \, \bSigma_{0,n,k}(r)
  \nonumber \\
&  - \int_0^t \! ds \, \int_0^s \! dr \, \ad{r}(\biV_{\varepsilon,n} - \biV_n)
  \, \tfrac{d}{ds} \bSigma_{0,n,k}(s) \, . 
\end{align}
As was shown for $k = 1$, the norm of the integral
$\int_0^s \! dr \, \ad{r}(\biV_{\varepsilon,n} - \biV_n)$ vanishes for
small $\varepsilon$ in norm, uniformly in $s$. 
Moreover $s \mapsto \bSigma_{0,n,k}(s)$ and its derivative are
uniformly bounded for $s$ in compact sets. Therefore, also the 
integral in the second lime
vanishes in norm for small $\varepsilon$,
uniformly in $t$. Hence
$\varepsilon \rightarrow \bSigma_{\varepsilon, k+1}(t)$
converges uniformly to $\bSigma_{0, k+1}(t)$, which completes 
the proof.


\vspace*{5mm}
\noindent {\Large \bf Acknowledgment} \\[1mm]
I would like to thank Teun van Nuland for comments during the
early stages of this project.

\end{document}